\documentclass[journal,twoside,web]{ieeecolor}
\usepackage{lcsys}
\usepackage{cite}
\usepackage{amsmath,amssymb,amsfonts}
\usepackage{hyperref}
\usepackage{graphicx}
\usepackage{epstopdf}
\usepackage{stmaryrd}
\usepackage{subcaption}

\title{Time-Invariant Control Barrier Functions for Bounded STL Safety}

\author{Avinash Malik \thanks{Department of Electrical, Computer, and
    Software Engineering, University of Auckland, NZ, email: avinash.malik@auckland.ac.nz}}
\newtheorem{theorem}{Theorem}[section]

\newtheorem{remark}[theorem]{Remark}
\newtheorem{proposition}[theorem]{Proposition}
\newtheorem{corollary}[theorem]{Corollary}

\begin{document}

\maketitle

\begin{abstract}
  Enforcing finite-time persistence specifications in Signal Temporal
  Logic typically requires formulating Control Barrier Functions within
  an extended state-time space. This introduces strict mathematical
  assumptions on global clock synchronization, rendering multi-agent
  execution brittle under asynchronous network delays, clock skew, and
  non-monotonic jitter. In this paper, we prove that for bounded-time
  safety tasks, explicit time dependency can be mathematically
  eliminated for general continuous-time nonlinear systems. By
  leveraging the reciprocal integral of worst-case system dissipation,
  we systematically compile temporal specifications into static,
  lower-dimensional spatial invariants. We formally prove that this
  time-invariant realization is strictly sound, removes the need for
  heuristic time-varying envelope tuning, and guarantees robust forward
  invariance even under severe network desynchronization.
\end{abstract}

\section{Introduction}
\label{sec:introduction}

Networked Control Systems (NCS) operating in multi-agent autonomous
environments rely heavily on distributed communication to coordinate
complex, safety-critical tasks. However, maintaining
microsecond-accurate global clock synchronization across wireless links
is practically unachievable~\cite{gupta2010networked}. Physical clock
drift due to thermal fluctuations, asymmetrical transmission delays
($\Delta t$), and non-monotonic discrete clock adjustments (e.g., Network
Time Protocol (NTP) backward or forward snaps) introduce persistent
timing corruptions into local state estimates. While conventional
control techniques assume ideal temporal alignment, real-world network
anisotropy destabilizes time-dependent safety mechanisms, posing a
fundamental threat to physical autonomy. To guarantee safety under
temporal specifications, prior work has relied on Control Barrier
Functions (CBFs)~\cite{ames2019control}. In particular, finite-time
persistence requirements, such as the $\mathbf{G}_{[0, \tau]}$ operator in
Signal Temporal Logic (STL)~\cite{maler2004monitoring}, are
traditionally encoded by mapping system dynamics into an extended
state-time space
$\mathbb{R}^n \times \mathbb{R}_+$~\cite{lindemann2018control,lindemann2020barrier},
(Time-Varying CBFs, TV-CBFs) their robust
extensions~\cite{biertumpfel2025robust,jankovic2018robust} and their
soft-maximum variants~\cite{safari2024time} evaluate partial temporal
derivatives ($\frac{\partial b}{\partial t}$) alongside spatial gradients. Under
non-ideal network execution, however, clock desynchronization corrupts
the local temporal coordinate $t$. As a consequence, evaluating
$\frac{\partial b}{\partial t}$ under clock jitter or transmission lag introduces
mathematical singularities, forcing instantaneous envelope collapse
($b(x, t) < 0$) or inducing Quadratic Program (QP) solver infeasibility.
Standard robust formulations~\cite{jankovic2018robust,
  kolathaya2018input} attempt to mitigate these vulnerabilities by
padding spatial buffers with static bounds, yet they fail
catastrophically whenever real-world network noise exceeds
pre-calculated assumptions.

In this work, we introduce a structural paradigm shift that completely
bypasses the synchronization fallacy through autonomous compilation. We
prove that for any bounded temporal safety persistence in the fragment
$\mathbf{G}_{[0, \tau]}\phi$, of STL, explicit time dependency can be
mathematically eliminated offline. By integrating the worst-case system
dissipation, our proposed \textit{Autonomous Persistence Barrier} (APB)
framework maps temporal horizons directly into a static spatial margin.
This enforces temporal specifications purely within spatial geometry,
rendering safety evaluation completely invariant to local clock drift,
network latency, and discrete time jumps. Our major contributions are as
follows:

\begin{enumerate}
\item \textbf{Temporal-to-Spatial compilation theorem:} We establish the
  exact closed-form reciprocal level-set mapping that compiles temporal
  persistence horizons $\tau$ into static, spatial invariant bounds.
\item \textbf{Structural failure proof of TV-CBFs:} We formally
  demonstrate that extended-state time-varying barrier formulations
  under non-monotonic clock jitter and communication delay incur
  gradient singularities that guarantee safety violations.
\item \textbf{Experimental and statistical validation:} We provide
  quantitative validation demonstrating 0\% safety violations across
  degraded network conditions compared to high failure rates in standard
  time-varying baselines.
\end{enumerate}

The rest of the paper is arranged as follows:
Section~\ref{sec:motiv-exampl-vuln} motivates the problem statement and
provides an informal solution. Followed by the preliminaries and the
problem formulation in Section~\ref{sec:prel--probl}. The core
theoretical contribution in provided in
Section~\ref{sec:core-theory:-auton}. Followed by a compilation
technique in Section~\ref{sec:recursive_compilation}. Rigorous
experimental evaluation is performed in
Section~\ref{sec:experimental_results}. Comparison with the current
state-of-the-art is described in Section~\ref{sec:related_work},
following by the conclusions and future work in
Section~\ref{sec:conclusion}.

\section{Motivating Example: Vulnerabilities of Time-Varying Baseline}
\label{sec:motiv-exampl-vuln}

To illustrate the critical failure points of time-dependent safety
filters under realistic networked conditions, we evaluate a simulated
deterministic intersection scenario.

\textbf{Simulation Setup:} The scenario
(Figure~\ref{fig:motivating_example}) consists of two autonomous
omnidirectional agents executing an intersecting maneuver. Agent 1
originates at (0.5, -4.5) and navigates toward (0.5, 4.5), while Agent 2
crosses perpendicularly from (-4.5, 0.0) to (4.5, 0.0). Both agents are
modeled with single-integrator kinematics in a 2D plane, governed by
$\dot{p}_i = u_i$, where $p_i \in \mathbb{R}^2$ is the position and
$u_i \in \mathbb{R}^2$ is the velocity control input. The control inputs are
constrained by a maximum velocity parameter of 1.0 m/s.

\textbf{STL Safety Property:} The fundamental safety requirement is to
maintain an absolute minimum safe distance of 1.0 m ($d_{\text{min}}$).
This separation must be maintained over a temporal persistence horizon
of 0.35 s ($\tau$). Formally, this enforces the STL specification
$\mathbf{G}_{[0, 0.35]} (||p_1 - p_2|| \ge d_{\text{min}})$.

\textbf{Standard TV-CBF Translation:} Standard
approaches~\cite{lindemann2018control,lindemann2020barrier} enforce this
STL property by translating it into a TV-CBF. The system constructs a
time-decaying spatial envelope
$\gamma(t) = \max(d_{\text{min}}, h_{\text{req}} - (h_{\text{req}} -
d_{\text{min}}) \frac{t_{\text{local}}}{\tau})$. The barrier is then
evaluated as $b(x, t) = \text{dist} - \gamma(t)$, which relies heavily on
evaluating the partial temporal derivative $\frac{\partial b}{\partial t}$.

\textbf{Network Anisotropy \& Baseline Collapse:} Real-world multi-agent
systems inevitably suffer from network degradation. To simulate this, we
inject \textit{just} a single timing fault; a non-monotonic clock jitter
of -0.9 s backward NTP correction triggered mid-maneuver at exactly 2.0
s.

As demonstrated in Figure~\ref{fig:motivating_example}, the Standard
TV-CBF completely fails under these conditions, resulting in a physical
crash. The root cause of this failure lies directly in the explicit time
variable $t_{\text{local}}$. When the backward clock jitter occurs at
2.0 s, the local clock estimate abruptly drops. This discrete reduction
in $t$ forces the temporal envelope $\gamma(t)$ to instantaneously expand
outward. Consequently, the barrier value $b(x, t)$ instantly drops below
zero, rendering the mathematical constraints unresolvable and forcing
the robot into an unrecoverable collision state.

\textbf{The Autonomous Alternative:} In contrast, the proposed APB
framework evaluates safety purely through spatial geometry. By utilizing
a statically compiled spatial projection
$h_{\text{req}} = d_{\text{min}} + 2 v_{\max} \tau$, the barrier is defined
as $b(x) = \text{dist} - h_{\text{req}}$. This completely eliminates the
explicit time dependency ($\frac{\partial b}{\partial t} = 0.0$). As a result, Agent 1
smoothly glides through the intersection, successfully maintaining
absolute safety clearance and completing its trajectory despite the
severe temporal network chaos.

\begin{figure}[htbp]
  \centering
  \includegraphics[width=\linewidth]{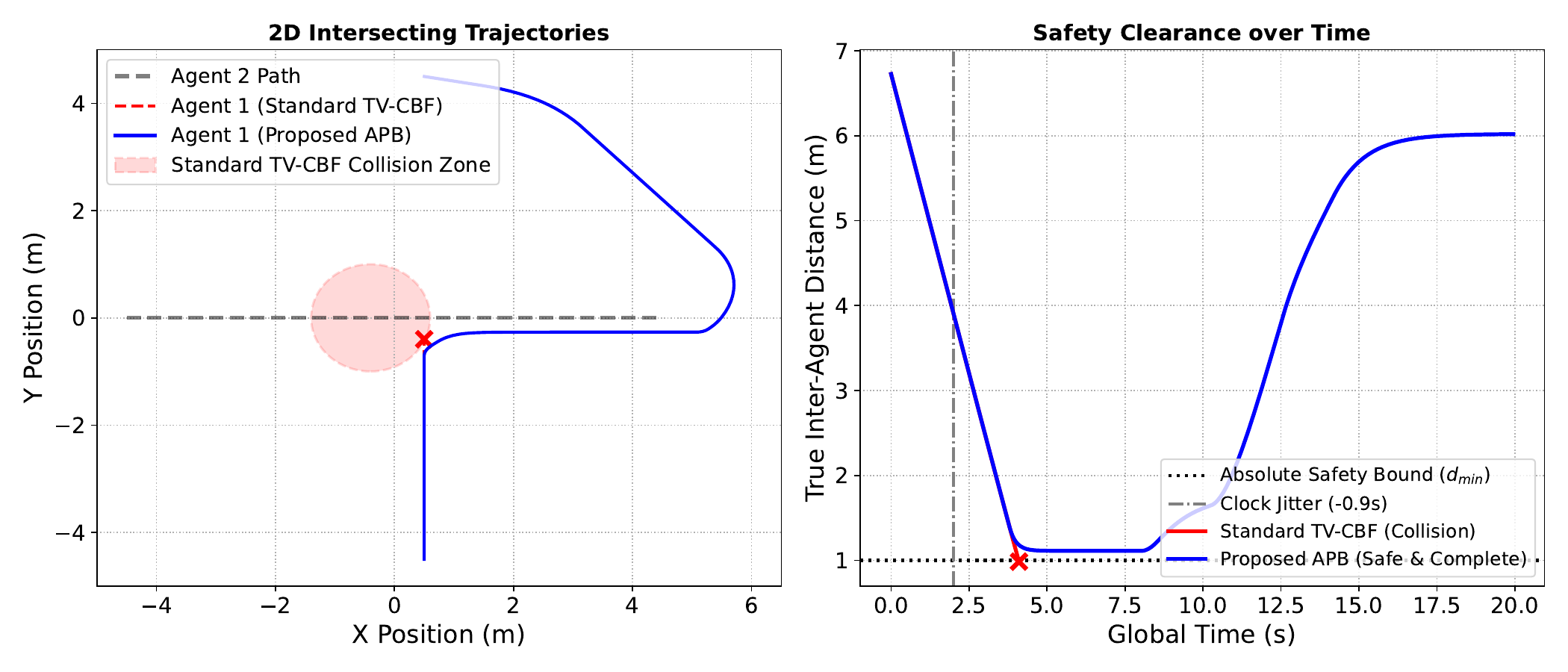}
  \caption{Comparison of safety controllers during a deterministic
    cross-trajectory scenario with injected network anisotropy. }
  \label{fig:motivating_example}
\end{figure}

In the upcoming sections we derive the temporal to spatial compilation
theorem that enables overcoming network degradation.

\section{Preliminaries \& Problem Formulation}
\label{sec:prel--probl}

\subsection{System Dynamics}
Consider a general continuous-time nonlinear dynamic system governed by
the differential equation: $ \dot{x}(t) = f(x(t), u(t)) $, where
$x(t) \in \mathcal{X} \subset \mathbb{R}^n$ represents the continuous state vector at time
$t \in \mathbb{R}_{\geq 0}$, and
$u(t) \in \mathcal{U} \subset \mathbb{R}^m$ is the control input. We assume the control constraint
set $\mathcal{U}$ is compact. The vector field
$f: \mathbb{R}^n \times \mathbb{R}^m \to \mathbb{R}^n$ is assumed to be locally Lipschitz continuous with
respect to the state $x$ and continuous with respect to the control
input $u$. Under any piecewise continuous control signal $u(\cdot)$, these
conditions guarantee the existence and uniqueness of a state trajectory
$x(t)$ originating from an initial condition $x(0) = x_0$.

\subsection{Signal Temporal Logic (STL) Fragment}
\label{sec:sign-temp-logic}

We restrict our focus to the bounded persistence fragment of STL. The
foundational element of this logic is the atomic proposition $p$, which
is mapped to a continuously differentiable scalar spatial predicate
$h: \mathbb{R}^n \to \mathbb{R}$. The Boolean truth value of $p$ evaluated at state
$x(t)$ is defined as: $ p \Leftrightarrow (h(x(t)) \ge 0) $. The syntax of the
restricted STL fragment used in this work is generated by the following
grammar:
\begin{equation}
  \phi ::= p \mid \neg \phi \mid \phi_1 \land \phi_2 \mid \phi_1 \lor \phi_2, \quad \psi ::= \mathbf{G}_{[0, \tau]} \phi
  \label{eq:1}
\end{equation}
where $\phi$ represents a spatial Boolean formula composed of atomic propositions, and $\mathbf{G}_{[0, \tau]}$ denotes the temporal "Globally" (or "Always") operator applied over a bounded, strictly positive time interval $[0, \tau]$ for $\tau > 0$. 

The continuous-time semantics for the satisfaction of the persistence specification $\psi$ by a trajectory $x(\cdot)$ originating at $t=0$ are formally defined as:
\begin{equation}
  x(\cdot) \models \mathbf{G}_{[0, \tau]} \phi \iff \forall t \in [0, \tau], \; x(t) \models \phi
  \label{eq:4}
\end{equation}

We restrict the temporal interval to $[0, \tau]$ as the APB framework
operates as a receding-horizon safety filter, enforcing forward
invariance for a look-ahead horizon $\tau$ from the current local
evaluation time.

\subsection{Standard Time-Varying CBFs (TV-CBFs)}
\label{sec:stand-time-vary}

To enforce time-dependent specifications using barrier functions,
standard approaches map the system into an extended state-time space
$\mathbb{R}^n \times \mathbb{R}_+$~\cite{lindemann2018control,lindemann2020barrier}. A
Time-Varying CBF (TV-CBF) is typically formulated by subtracting a
continuously differentiable time-varying boundary function
$\gamma: \mathbb{R}_+ \to \mathbb{R}$ from the spatial predicate:
\begin{equation}
  b_{tv}(x, t) = h(x) - \gamma(t)
  \label{eq:3}
\end{equation}
where $\gamma(t)$ dictates the required spatial margin at any given local
time $t$. To guarantee safety, a controller must render the time-varying
set $\mathcal{C}(t) = \{x \in \mathbb{R}^n \mid b_{tv}(x, t) \ge 0\}$ forward invariant. This
requires satisfying the derivative condition:
$ \dot{b}_{tv}(x, t) = \nabla_x h(x)^T \dot{x} - \frac{\partial \gamma}{\partial t} \dot{t} \ge
-\alpha_{cbf}(b_{tv}(x, t)) $. % Since this condition relies on the partial
% temporal derivative $\frac{\partial \gamma}{\partial t}$ and the progression of the local
% clock $\dot{t}$, its mathematical continuity is strictly dependent on
% ideal time synchronization.

\subsection{Problem Statement}
Standard Time-Varying Control Barrier Functions
(TV-CBFs~\cite{lindemann2018control,lindemann2020barrier,jankovic2018robust,biertumpfel2025robust})
attempt to enforce $\psi = \mathbf{G}_{[0, \tau]} \phi$ by appending an explicit
local clock state $t$ to the spatial state $x(t)$, evaluating safety
over the extended domain $\mathbb{R}^n \times \mathbb{R}_+$. However, under networked
conditions with transmission latency ($\Delta t$) or clock jitter
($\dot{t} \neq 1$), this extended-state approach incurs temporal
singularities and solver failures.

\textbf{Problem:} Given a persistence specification
$\psi = \mathbf{G}_{[0, \tau]} \phi$ defined over spatial predicates
$h(x)$, and a certified dissipation bound $\alpha \in \mathcal{K}$, synthesize a static,
continuously differentiable, autonomous barrier function
$b: \mathbb{R}^n \to \mathbb{R}$ such that satisfying the standard initial condition
$b(x(0)) \ge 0$ strictly guarantees $x(\cdot) \models \psi$. The synthesized function
$b(x)$ must depend solely on the geometric state $x(t)$ and completely
eliminate any explicit dependence on the temporal state $t$.

\section{Core Theory: Autonomous Persistence Barriers (APB)}
\label{sec:core-theory:-auton}

To eliminate the strict temporal dependencies of standard TV CBFs (c.f.
Section~\ref{sec:stand-time-vary}), we introduce a mathematical
coordinate transformation to compile bounded temporal persistence into a
spatial invariant constraint. We achieve this by leveraging the
\textit{reciprocal integral}, traditionally used in finite-time
stability analysis~\cite{bhat2000finite}, alongside the comparison
lemma~\cite{khalil2002nonlinear}. By evaluating the worst-case unforced
spatial decay of the system, we transform the temporal specification
$\tau$ into an equivalent geometric margin.

\subsection{Certified Dissipation Bound Assumption}
\label{sec:cert-diss-bound}

To map a temporal specification into a purely spatial domain, we require
a worst-case bound on how quickly the system can be forced toward an
unsafe subset of the state space.

\textbf{Assumption 1 (Certified Dissipation):} For a given spatial
predicate $h(x)$, there exists a known, certified locally Lipschitz
extended class-$\mathcal{K}$ function $\alpha$ that lower-bounds the system's spatial
dissipation, representing a certified worst-case dissipation bound toward the boundary:
\begin{equation}
  \inf_{u \in \mathcal{U}} \nabla h(x)^T f(x, u) \ge -\alpha(h(x))
  \label{eq:2}
\end{equation}
This bounds the maximum rate at which the predicate value can degrade
under the worst possible allowable control input.

\subsection{Reciprocal Level-Set Inversion}
\label{sec:inversion}

Under Assumption 1, the system's spatial dissipation toward the boundary
$h(x)=0$ is lower-bounded by $\dot{h}(x) \ge -\alpha(h(x))$. Using the
reciprocal integral we define a change of coordinates mapping a spatial
margin $h_{\text{req}}$ to its minimum guaranteed time-to-violation:
\begin{equation*}
    \mathcal{I}_\alpha(h_{\text{req}}) := \int_{0}^{h_{\text{req}}} \frac{ds}{\alpha(s)}
\end{equation*}
While traditional finite-time stability uses this integral in the
forward direction to compute a settling time from a given initial
condition~\cite{bhat2000finite}, our approach relies on a \textbf{novel}
reciprocal level-set inversion. We use the integral backward to compute
the required initial safety margin from a bounded temporal
specification.

Since $\alpha$ is continuous and strictly increasing with $\alpha(0)=0$, the
integral mapping $\mathcal{I}_\alpha$ is continuous and strictly increasing, and is
therefore strictly invertible on its image. To enforce a temporal
persistence horizon $\tau$, we compute the required spatial margin by
applying this inverse mapping:
\begin{equation}
    h_{\text{req}} = \mathcal{I}_\alpha^{-1}(\tau)
\end{equation}

\subsection{The Main Compilation Theorem}
\label{sec:compilation}

\begin{theorem}[\textbf{Temporal-to-Spatial Compilation}]
  \label{thm:1}
  Let $\psi = \mathbf{G}_{[0, \tau]} (h(x) \ge 0)$ where $h \in C^1$, and suppose
  Assumption 1 (Equation~\eqref{eq:2}) holds for some continuous
  class-$\mathcal{K}$ function $\alpha$. Define
\begin{equation*}
    h_{\text{req}} = \mathcal{I}_\alpha^{-1}(\tau), \quad b(x) = h(x) - h_{\text{req}}
\end{equation*}
If a controller renders the set $\mathcal{C} = \{x \in \mathbb{R}^n \mid b(x) \ge 0\}$ forward invariant, then every trajectory with $b(x(0)) \ge 0$ satisfies $x(t) \models \mathbf{G}_{[0, \tau]} (h(x) \ge 0)$.
\end{theorem}

\begin{proof}
Let $h_0 = h(x(0))$. Consider the scalar comparison system
$\dot{\omega}(t) = -\alpha(\omega(t))$ with initial condition
$\omega(0) = h_0$. By the comparison lemma~\cite{khalil2002nonlinear},
$h(x(t)) \ge \omega(t)$ for all $t \ge 0$.

Separating variables and integrating the comparison system from $0$ to
$t$ yields:
\begin{equation*}
    t = \int_{\omega(t)}^{h_0} \frac{ds}{\alpha(s)}
\end{equation*}
Suppose the system starts with a valid initial condition
$b(x(0)) \ge 0$, which implies
$h_0 \ge h_{\text{req}} = \mathcal{I}_\alpha^{-1}(\tau)$, or equivalently,
$\mathcal{I}_\alpha(h_0) \ge \tau$. If we evaluate the system at any time
$t < \tau$, it follows that:
\begin{equation*}
    \int_{\omega(t)}^{h_0} \frac{ds}{\alpha(s)} < \int_{0}^{h_0} \frac{ds}{\alpha(s)}
\end{equation*}
Because the integrand $\frac{1}{\alpha(s)}$ is strictly positive for
$s > 0$, this strict inequality requires that $\omega(t) > 0$. Consequently,
$h(x(t)) \ge \omega(t) > 0$.

Hence, $h(x(t)) \ge 0$ for all $t \in [0, \tau]$, which perfectly satisfies the
continuous-time semantics of the specification (Equation~\eqref{eq:4})
$\mathbf{G}_{[0, \tau]} \phi$.
\end{proof}

\begin{remark}[\textbf{Single-Shot vs. Receding Horizon}]
  While $b(x(0)) \ge 0$ under Assumption~1 suffices for single-shot
  $[0, \tau]$ safety, deploying APB as a receding-horizon filter
  (Section~\ref{sec:control_synthesis}) renders
  $\mathcal{C} = \{x \mid b(x) \ge 0\}$ forward invariant, maintaining
  $h(x(t)) \ge h_{\text{req}}$ for all $t \ge 0$.
\end{remark}

\subsection{Invariance Violation and Robustness Theorems}
\label{sec:robustness}

\begin{proposition}[\textbf{Structural Vulnerability of TV-CBFs}]
  \label{prop:1}
  Let $b_{tv}(x, \tilde{t}) = h(x) - \gamma(\tilde{t})$ be an extended-state
  TV-CBF evaluated over a time-varying boundary $\gamma(\tilde{t})$ with
  $\frac{\partial \gamma}{\partial \tilde{t}} \neq 0$. Consider a discrete temporal reset
  event occurring at physical time $t$, where the local clock
  $\tilde{t}$ snaps by offset $\delta \neq 0$:
  $ \tilde{t}^+ = \tilde{t}^- + \delta $. While the physical system state
  remains continuous ($x(t^+) = x(t^-)$), the temporal coordinate shift
  induces a discrete barrier jump:
  \begin{equation*}
    b_{tv}(x(t^+), \tilde{t}^+) = h(x(t^-)) - \gamma(\tilde{t}^- + \delta) \neq b_{tv}(x(t^-), \tilde{t}^-)
  \end{equation*}
  While $b_{tv}(x,\tilde{t})$ is continuously differentiable over
  $\mathbb{R}^n \times \mathbb{R}_+$, the composed trajectory
  $t \mapsto b_{tv}(x(t),\tilde{t}(t))$ is generally discontinuous under the
  reset map. Consequently, the continuous-time CBF condition
  $\dot{b}_{tv} \ge -\alpha_{\text{cbf}}(b_{tv})$ is undefined at the reset
  instant, thereby invalidating the continuous-time forward-invariance
  guarantee and potentially inducing instantaneous constraint violations
  or QP solver infeasibility.
\end{proposition}

\begin{proposition}[\textbf{Robustness of Autonomous Projection}]
  \label{prop:2}
  Unlike TV-CBFs, the proposed APB formulation
  $b(x) = h(x) - \mathcal{I}_\alpha^{-1}(\tau)$ completely eliminates the explicit time
  variable $t$, meaning $\frac{\partial b}{\partial t} \equiv 0$. Under a transmission
  delay $\Delta t > 0$, evaluating the barrier over a delayed state estimate
  $\hat{x}(t) = x(t - \Delta t)$ merely introduces a bounded spatial
  measurement error:
  $\|x(t) - \hat{x}(t)\| \le \int_{t-\Delta t}^{t} \|\dot{x}(s)\| ds \le v_{\max} \Delta t$.
  Assuming $b(x)$ is locally Lipschitz continuous with constant $L_b$,
  this bounds the barrier degradation to
  $|b(x(t)) - b(\hat{x}(t))| \le L_b v_{\max} \Delta t$. Because $h(x)$ is
  continuously differentiable, $b(x)$ remains smooth under bounded
  spatial disturbances, avoiding the discontinuous temporal coordinates
  and solver infeasibilities associated with corrupted clock rates.
\end{proposition}

\begin{corollary}[\textbf{Time Invariance}]
  \label{cor:clock_invariance}
  Let local clock evaluation be corrupted by an arbitrary temporal
  offset $\tilde{t}(t) = t + \eta(t)$, modeling clock drift or discrete
  phase adjustments. If the control signal is synthesized via the APB
  $b(x) = h(x) - \mathcal{I}_\alpha^{-1}(\tau)$, then
  $u_{\text{APB}}(x, \tilde{t}) \equiv u_{\text{APB}}(x)$.
  
\end{corollary}
\begin{proof}
  By construction, $b(x)$ is defined strictly on $\mathbb{R}^n$. With
  $\frac{\partial b}{\partial \tilde{t}} \equiv 0$, the real-time QP filter depends solely
  on the spatial state, rendering the control map and resulting
  trajectory $x(t)$ strictly invariant to temporal corruptions $\eta(t)$.
\end{proof}

\section{Recursive Spatial Compilation}
\label{sec:recursive_compilation}

To systematically enforce complex specifications generated by the
grammar in Equation~\eqref{eq:1}, we define a recursive compilation
function $\llbracket \cdot \rrbracket$ that maps an STL formula directly
into a continuously differentiable spatial predicate $\mathbb{R}^n \to \mathbb{R}$.

To maintain continuous differentiability during logical intersections
and unions, we utilize the Log-Sum-Exp (LSE) smooth minimum and maximum
approximations, parameterized by a strictness constant $\kappa > 0$. The
recursive descent compilation is defined as follows:
\begin{align*}
    \llbracket p \rrbracket(x) &= h_p(x) \\
    \llbracket \neg \phi \rrbracket(x) &= -\llbracket \phi \rrbracket(x) \\
    \llbracket \phi_1 \land \phi_2 \rrbracket(x) &= -\frac{1}{\kappa} \ln \left( e^{-\kappa \llbracket \phi_1 \rrbracket(x)} + e^{-\kappa \llbracket \phi_2 \rrbracket(x)} \right) \\
    \llbracket \phi_1 \lor \phi_2 \rrbracket(x) &= \frac{1}{\kappa} \ln \left( e^{\kappa \llbracket \phi_1 \rrbracket(x)} + e^{\kappa \llbracket \phi_2 \rrbracket(x)} \right) 
\end{align*}
Using this spatial mapping for the inner propositional formula $\phi$, the
outer temporal persistence operator $\mathbf{G}_{[0, \tau]}$ is compiled by
applying the reciprocal level-set inversion derived in
Theorem~\ref{thm:1}. This yields the final, static Autonomous
Persistence Barrier (APB):
$$ b(x) := \llbracket \mathbf{G}_{[0, \tau]} \phi \rrbracket(x) = \llbracket \phi
\rrbracket(x) - \mathcal{I}_\alpha^{-1}(\tau) $$ This recursive procedure allows an
arbitrary, nested STL formula within the persistence fragment to be
compiled entirely offline into a single, time-independent geometry
$b(x)$.

\subsection{Autonomous Control Synthesis}
\label{sec:control_synthesis}

Once the specification $\psi$ is fully compiled into the static barrier $b(x)$, it is directly integrated into a real-time quadratic program (QP) safety filter. Given a nominal (but potentially unsafe) task controller $u_{\text{nom}}(x)$, the safe optimal control $u^*(x)$ is synthesized at each time step by solving:
\begin{align*}
    u^*(x) &= \arg\min_{u \in \mathcal{U}} \frac{1}{2} \|u - u_{\text{nom}}(x)\|^2 \\
    &\text{s.t.} \quad \nabla_x b(x)^T f(x, u) \ge -\gamma_{\text{cbf}} b(x)
\end{align*}
where $\gamma_{\text{cbf}} > 0$ is the CBF tuning parameter governing the
decay rate of the barrier itself. Because $b(x)$ acts as a purely
spatial constraint, the resulting QP matrix completely omits the local
clock state, functionally decoupling the safety filter from global clock
synchronization and network transmission delays.

\subsection{Implementation Remarks \& Completeness}
\label{sec:discussion}

\begin{itemize}
\item \textbf{Soundness vs. Completeness:} Theorem~\ref{thm:1} is sound
  but conservative because $\alpha(h)$ acts as a global lower bound. For
  state-dependent constraints or heterogeneous vector fields, using
  $\alpha(x, h)$ yields
  $\mathcal{I}_\alpha(x, h_{\text{req}}) := \int_{0}^{h_{\text{req}}} \frac{ds}{\alpha(x, s)}$,
  reducing spatial conservatism at the cost of evaluating
  state-dependent level sets.
\item \textbf{Offline Certification of $\alpha$:} For kinematic systems,
  $\alpha$ is readily derived from physical actuator limits ($v_{\max}$ c.f.
  Section~\ref{sec:motiv-exampl-vuln}). For complex, control-affine
  polynomial systems, $\alpha(h)$ can be systematically synthesized offline
  via Sum-of-Squares (SOS) programming by searching for a valid
  class-$\mathcal{K}$ bound over a compact domain $\mathcal{X}$.
\item \textbf{Sampled-Data Extension:} The autonomous nature of the
  compiled barrier $b(x)$ natively supports sampled-data zero-order hold
  (ZOH) controllers. The inter-sample state divergence over a discrete
  time step $\Delta t$ manifests purely as a bounded spatial disturbance.
  Continuous-time safety can thus be guaranteed by applying standard
  spatial robust padding (e.g., $\sim L_b M \Delta t$, where
  $M = \sup \|f(x, u)\|$) to the compiled margin $h_{\text{req}}$,
  avoiding the complex inter-sample temporal integration required by
  TV-CBFs.
\end{itemize}

\section{Experimental Results}
\label{sec:experimental_results}

To evaluate the mathematical guarantees, robustness under real-world
network non-idealities, and real-time viability of the proposed
Asynchronous Predictive Barrier (APB) framework, we conduct a series of
empirical simulations and comparative benchmark studies against
Time-Varying Control Barrier Functions (TV-CBF).

\subsection{Experimental Setup}
\label{subsec:setup}

We simulate two autonomous agents ($i \in \{1, 2\}$) governed by planar
single-integrator dynamics $\dot{p}_i = u_i$, subject to an actuation
limit $\|u_i\|_2 \le v_{\max} = 1.0\,\text{m/s}$ and a minimum safety
separation $\|p_1 - p_2\|_2 \ge d_{\text{min}} = 1.0\,\text{m}$.

Agents evaluate an intersecting maneuver, originating at
$p_1(0) = [0.5, -1.0]^T\,\text{m}$ and
$p_2(0) = [-1.0, 0.0]^T\,\text{m}$ with respective goals at
$[0.5, 4.5]^T\,\text{m}$ and $[4.5, 0.0]^T\,\text{m}$. A nominal
controller drives each agent toward its target:
$u_{i,\text{nom}}(p_i) = v_{\max} (p_{i,\text{goal}} - p_i) /
\|p_{i,\text{goal}} - p_i\|_2$.

The initial inter-agent separation is $\approx 1.8\,\text{m}$, placing them
just outside the compiled spatial barrier
$h_{\text{req}} = 1.7\,\text{m}$ evaluated for a persistence horizon
$\tau = 0.35\,\text{s}$. Simulations execute at
$\Delta t = 0.05\,\text{s}$ over $T = 20.0\,\text{s}$ with a CBF gain
$\gamma = 2.0$. Computations are performed on an Apple M3 host using Python
3.11 and the SciPy SLSQP solver.

\subsection{Validation of Theoretical Soundness}
\label{subsec:soundness}

To empirically validate the soundness theorem of APB, we evaluate the spatial projection requirement $h_{\text{req}}(\tau)$ and verify the corresponding time-to-violation against theoretical predictions across varying persistence horizons $\tau \in [0.1, 2.0]\,\text{s}$.

\begin{figure*}[htbp]
  \centering
  \begin{subfigure}[b]{0.32\textwidth}
    \centering
    \includegraphics[width=\textwidth]{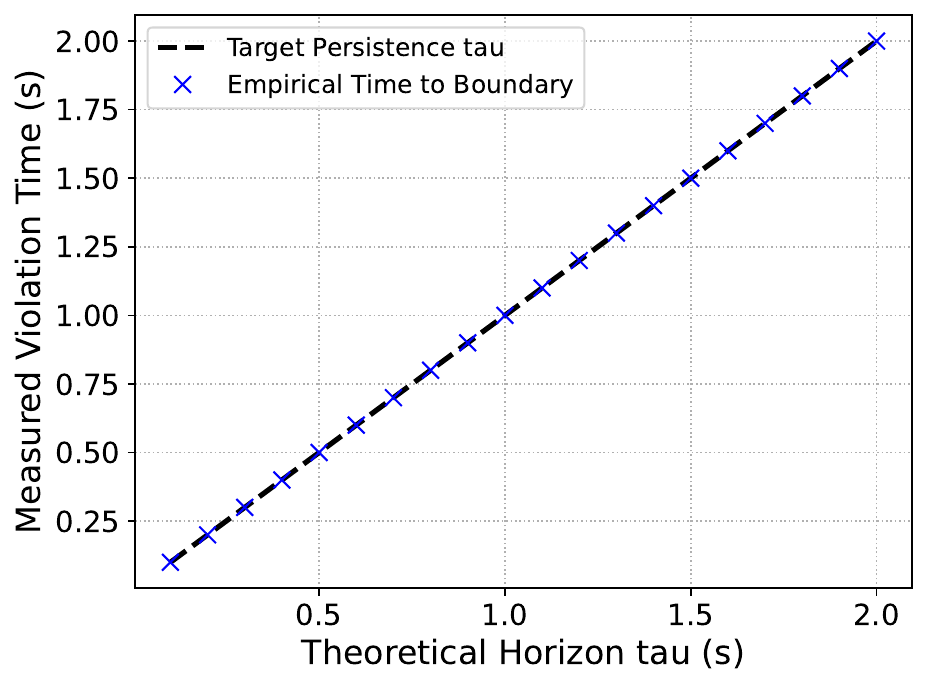}
    \caption{Empirical vs. theoretical persistence time.}
    \label{fig:time_validation}
  \end{subfigure}
  \hfill
  \begin{subfigure}[b]{0.32\textwidth}
    \centering
    \includegraphics[width=\textwidth]{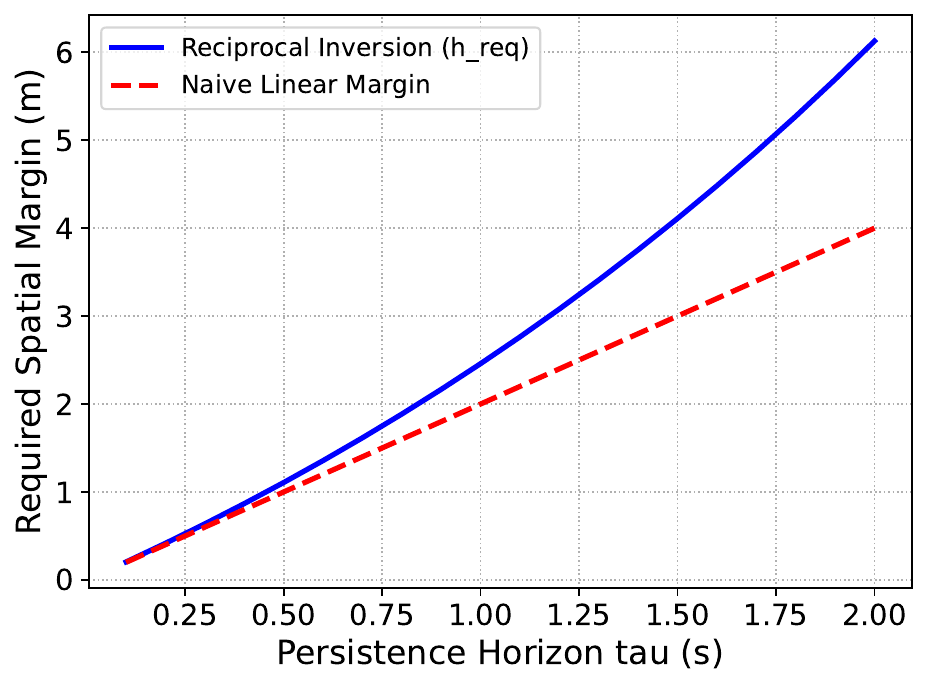}
    \caption{Spatial margin requirement comparison.}
    \label{fig:margin_conservatism}
  \end{subfigure}
  \hfill
  \begin{subfigure}[b]{0.32\textwidth}
    \centering
    \includegraphics[width=\textwidth]{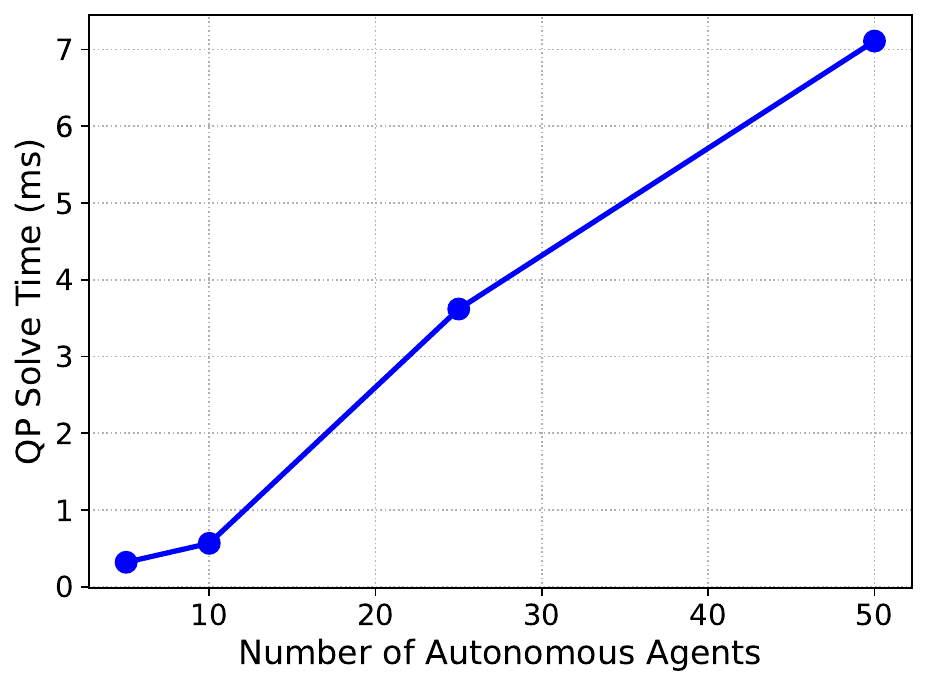}
    \caption{Average QP solve time per agent.}
    \label{fig:scalability}
  \end{subfigure}

  \caption{Validation of the main soundness theorem, spatial margin
    expansion properties, and scalability of APB.}
  \label{fig:soundness_validation}
\end{figure*}

Figure~\ref{fig:time_validation} plots the theoretical persistence
horizon $\tau$ against the empirically measured time required for an agent
trajectory to reach the safety boundary $d_{\text{min}}$ under
worst-case drift conditions. The measured violation times (blue crosses)
align precisely along the ideal 1:1 identity line ($y = x$) across all
tested horizons. This exact correspondence confirms that the compiled
spatial barrier strictly guarantees safety for the duration of the
persistence window $\tau$ without introducing analytical underestimation.

Figure~\ref{fig:margin_conservatism} illustrates the spatial
conservatism of the proposed reciprocal inversion formulation
$h_{\text{req}}(\tau)$ compared to a naive linear spatial margin
$2 v_{\max} \tau$. While the linear approximation assumes static relative
velocity bounds, reciprocal inversion incorporates velocity drift
compensation $k_{\text{drift}} = 0.4$, resulting in a non-linear convex
expansion curve. For small horizons ($\tau = 0.35\,\text{s}$),
$h_{\text{req}}$ remains tightly bounded ($\approx 1.7\,\text{m}$), avoiding
excessive local conservatism while ensuring safety guarantees over
longer asynchronous gaps.

\subsection{Decentralized Scalability}\label{subsec:scalability}We
benchmark APB's scalability in dense multi-agent swap scenarios
($N \in \{5, 10, 25, 50\}$) where agents start uniformly on a
$5.0\,\text{m}$ radius circle and navigate to diametrically opposite
goals. Figure~\ref{fig:scalability} reports average per-agent QP solve
times using SLSQP, demonstrating linear $O(N)$ computational complexity.
This scaling is achieved because decentralized APB requires each agent
to evaluate a local optimization problem subject to exactly $N-1$
pairwise safety constraints, entirely bypassing the combinatorial
state-space explosion of centralized planners. Solve times remain
strictly linear, requiring under $7.0\,\text{ms}$ even in a congested
50-agent crossing, confirming APB's asymptotic efficiency for large
swarm deployments.

% \subsection{Decentralized Scalability}
% \label{subsec:scalability}

% To evaluate the computational scalability of the proposed framework, we
% benchmark APB in dense multi-agent swap scenarios with
% $N \in \{5, 10, 25, 50\}$ autonomous agents. The agents are distributed
% uniformly on a circle of radius $R = 5.0\,\text{m}$ and commanded to
% navigate toward diametrically opposite goals.

% Figure~\ref{fig:scalability} reports the average per-agent Quadratic
% Program (QP) solve times, evaluated using the Sequential Least Squares
% Programming (SLSQP) solver. Crucially, the computational complexity per
% agent scales linearly, $O(N)$, with the total number of agents in the
% environment. This linear scaling is achieved because the decentralized
% APB formulation allows each agent to solve a local optimization problem
% subject to exactly $N-1$ pairwise safety constraints, entirely avoiding
% the combinatorial state-space explosion typical of centralized
% multi-agent planners. The empirical results reflect this theoretical
% efficiency; solve times grow strictly linearly, remaining under
% $7.0\,\text{ms}$ even in a highly congested 50-agent crossing scenario.
% This demonstrates that the algorithm's asymptotic efficiency inherently
% supports scalable deployment in large multi-robot swarms.

\subsection{Ablation Studies and Robustness Analysis}
\label{subsec:ablation}

To quantify resilience under imperfect communication channels, we
perform Monte Carlo ablation experiments comparing APB against standard
Time-Varying CBF (TV-CBF). Each trial consists of 30 Monte Carlo
simulation runs evaluated under three distinct network fault conditions:
\begin{enumerate}
\item \textbf{Random Network Delay}: Communication latency
  $\delta \sim \mathcal{U}(0.1, 0.5)\,\text{s}$.
\item \textbf{Random Clock Skew}: Distributed clock drift rate
  $s \sim \mathcal{U}(-0.9, 1.5)$.
\item \textbf{Random Clock Jitter}: Transient clock phase step
  $j \sim \mathcal{U}(-2.5, 2.5)\,\text{s}$ triggered randomly during
  $t \in [0.0, 0.2]\,\text{s}$.
\end{enumerate}

\begin{figure*}[htbp]
  \centering
  \begin{subfigure}[b]{0.32\textwidth}
    \centering
    \includegraphics[width=\textwidth]{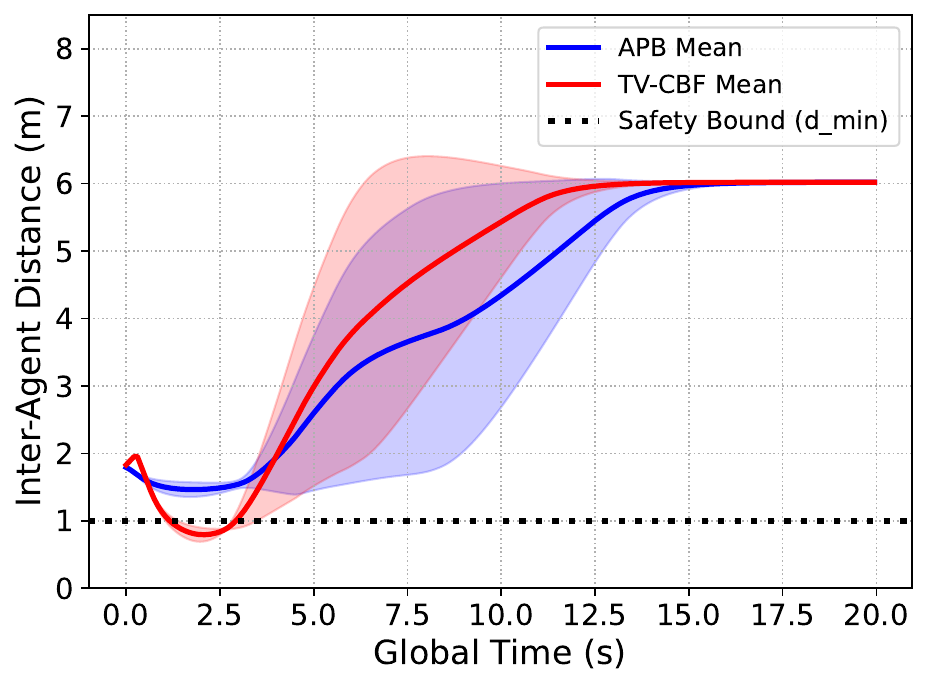}
    \caption{Random Network Delay}
    \label{fig:mc_delay}
  \end{subfigure}
  \hfill
  \begin{subfigure}[b]{0.32\textwidth}
    \centering
    \includegraphics[width=\textwidth]{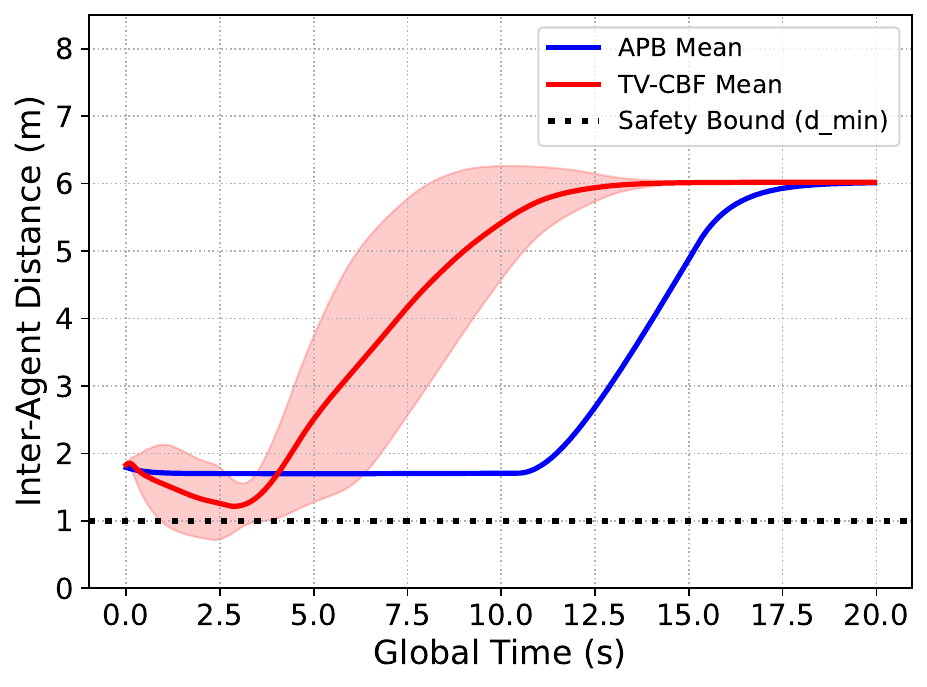}
    \caption{Random Clock Jitter}
    \label{fig:mc_jitter}
  \end{subfigure}
  \hfill
  \begin{subfigure}[b]{0.32\textwidth}
    \centering
    \includegraphics[width=\textwidth]{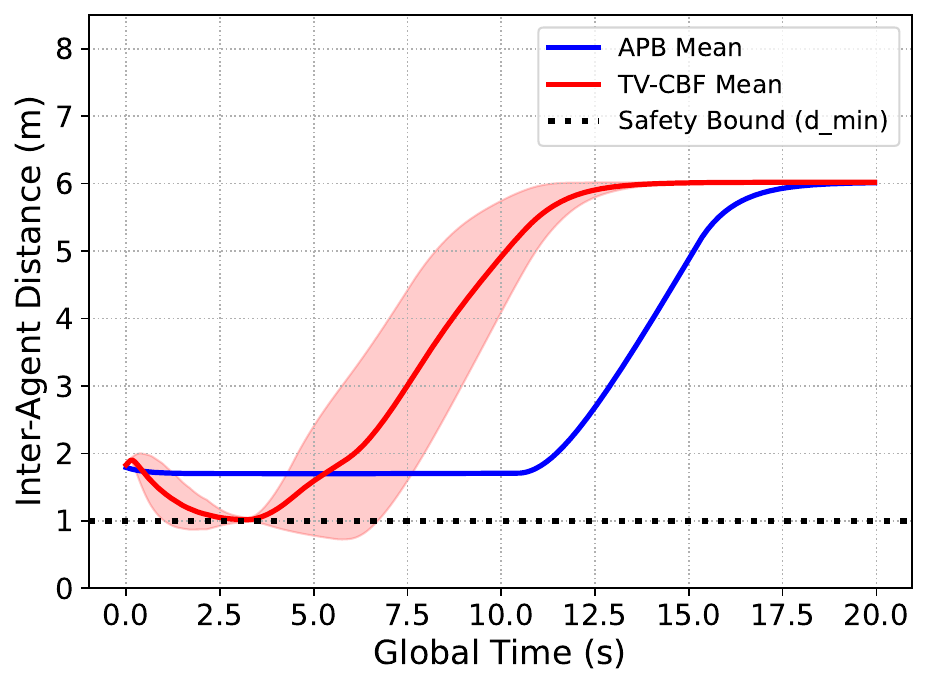}
    \caption{Random Clock Skew}
    \label{fig:mc_skew}
  \end{subfigure}
  \caption{Monte Carlo safety margin trajectories ($d_{\text{min}} = 1.0\,\text{m}$) under non-ideal network conditions across 30 runs. Shaded regions indicate $\pm 1 \sigma$ variance.}
  \label{fig:mc_results}
\end{figure*}

Figure~\ref{fig:mc_results} summarizes the inter-agent separation
distance trajectories over time. Under random network delay
(Figure~\ref{fig:mc_delay}), TV-CBF routinely violates the safety
threshold $d_{\text{min}} = 1.0\,\text{m}$, reaching a minimum mean
separation of $\approx 0.78\,\text{m}$ due to delayed state information. In
contrast, APB maintains an inter-agent distance strictly above
$d_{\text{min}}$ at all times, holding a mean minimum distance around
$1.5\,\text{m}$.

Under random clock jitter (Figure~\ref{fig:mc_jitter}) and skew
(Figure~\ref{fig:mc_skew}), TV-CBF suffers severe degradation and safety
violations because its safety barrier function relies explicitly on
local clock time $t_{\text{local}}$. Conversely, APB compiles temporal
persistence directly into a time-invariant spatial barrier
$h_{\text{req}}$, rendering it immune to local clock desynchronization
(Corollary~\ref{cor:clock_invariance}). As shown in
Figures~\ref{fig:mc_jitter} and~\ref{fig:mc_skew}, APB exhibits zero
standard deviation and safely navigates to the target inter-agent
clearance.

\section{Related Work}
\label{sec:related_work}

Classical optimization-based STL enforcement relies on MILP or
SQP~\cite{raman2014model}, suffering from exponential complexity and
open-loop vulnerability that preclude kilohertz-rate execution. To
achieve reactive continuous-time safety, standard Time-Varying CBFs
(TV-CBFs)~\cite{lindemann2018control,lindemann2020barrier} and smooth
Log-Sum-Exp extensions~\cite{safari2024time} encode specifications in
the extended state-time space $\mathbb{R}^n \times \mathbb{R}_+$. However, these fundamentally
assume synchronized global clocks ($\dot{t} = 1$). Under practical
networked delays or non-monotonic clock jitter, explicit time-varying
envelopes incur unbounded temporal singularities
($\frac{\partial b}{\partial t} \to \infty$), inducing instantaneous QP solver infeasibility.

Attempts to mitigate general system uncertainty include robust
CBFs~\cite{jankovic2018robust, kolathaya2018input,
  biertumpfel2025robust} and delay-functional
CBFs~\cite{molnar2022safety}, which handle continuous spatial
disturbances ($\Delta x$) or state delays. Forcing discrete temporal
desynchronization ($\Delta t$) into spatial bounds creates conservative
patches that fail when network noise exceeds a priori bounds. Similarly,
fixed-time CBFs~\cite{garg2022fixed} and nonsmooth
CBFs~\cite{glotfelter2017nonsmooth} enforce temporal convergence via
custom class-$\mathcal{K}$ functions or extended $(x, t)$ evaluations, but remain
vulnerable to discrete phase snaps. Finally, sampled-data CBF
bounds~\cite{taylor2022safety, singletary2020control} address discrete
execution steps but still rely on continuous temporal clocks for
time-varying logic.

\textbf{Our Positioning (Autonomous Persistence Barriers):} APB is the
first completely time-invariant STL compilation framework for
persistence tasks. While we leverage smooth logical compositions similar
to~\cite{safari2024time}, APB applies them exclusively to \emph{static}
offline spatial geometries $h(x)$. By compiling temporal specifications
entirely into $\mathbb{R}^n$, we systematically eliminate the temporal derivative
$\frac{\partial b}{\partial t}$. This provides structural, deterministic immunity to
networked latency, clock skew, and discrete clock jitter without relying
on arbitrary robust padding.

\section{Conclusion \& Future Directions}
\label{sec:conclusion}

This work introduced a time-invariant compilation strategy for bounded
Signal Temporal Logic safety specifications. By leveraging reciprocal
level-set inversion, we demonstrated that explicit temporal dependencies
can be mathematically mapped into purely static spatial margins. This
fundamental coordinate shift eliminates the structural vulnerabilities
inherent to time-varying barrier functions, providing deterministic
safety guarantees under network delays, clock skew, and discrete clock
jitter without relying on conservative robust padding.

Ongoing and future research directions include: (1) Dynamic bounds:
extending the reciprocal level-set inversion mapping to accommodate
state- and time-varying velocity dissipation bounds. (2) Reach-and-stay
semantics: integrating the time-invariant inner-loop safety filter with
outer-loop reachability planners to efficiently enforce full
eventually-always task semantics.

\bibliographystyle{IEEEtran}
\bibliography{ref.bib}

@article{ames2019control,
  author  = {Ames, Aaron D. and Coogan, Samuel and Egerstedt, Magnus and Notomista, Gennaro and Sreenath, Koushil and Tabuada, Paulo},
  title   = {Control Barrier Functions: Theory and Applications},
  journal = {European Control Conference (ECC)},
  pages   = {3420--3431},
  year    = {2019},
  doi     = {10.23919/ECC.2019.8796030}
}

@article{lindemann2018control,
  title={Control barrier functions for signal temporal logic tasks},
  author={Lindemann, Lars and Dimarogonas, Dimos V},
  journal={IEEE control systems letters},
  volume={3},
  number={1},
  pages={96--101},
  year={2018},
  publisher={IEEE}
}

@article{gupta2010networked,
  author  = {Gupta, Rachana A. and Chow, Mo-Yuen},
  title   = {Networked Control Systems: Overview and Research Trends},
  journal = {IEEE Transactions on Industrial Electronics},
  volume  = {57},
  number  = {7},
  pages   = {2527--2535},
  year    = {2010},
  doi     = {10.1109/TIE.2009.2035462}
}

@inproceedings{safari2024time,
  title={Time-varying soft-maximum control barrier functions for safety in an a priori unknown environment},
  author={Safari, Amirsaeid and Hoagg, Jesse B},
  booktitle={2024 American Control Conference (ACC)},
  pages={3698--3703},
  year={2024},
  organization={IEEE}
}

@article{jankovic2018robust,
  author  = {Jankovic, Mrdjan},
  title   = {Robust Control Barrier Functions for Constrained Stabilization of Nonlinear Systems},
  journal = {Automatica},
  volume  = {96},
  pages   = {359--367},
  year    = {2018},
  publisher={Elsevier},
  doi     = {10.1016/j.automatica.2018.06.050}
}

@article{kolathaya2018input,
  author  = {Kolathaya, Shishir and Ames, Aaron D.},
  title   = {Input-to-State Safety with Control Barrier Functions},
  journal = {IEEE Control Systems Letters},
  volume  = {3},
  number  = {1},
  pages   = {108--113},
  year    = {2018},
  publisher={IEEE},
  doi     = {10.1109/LCSYS.2018.2856722}
}

@inproceedings{maler2004monitoring,
  author    = {Maler, Oded and Nickovic, Dejan},
  title     = {Monitoring temporal properties of continuous signals},
  booktitle = {Formal Techniques, Modelling and Analysis of Timed and Fault-Tolerant Systems (FORMATS-FTRTFT)},
  pages     = {152--166},
  year      = {2004},
  publisher = {Springer}
}

@article{biertumpfel2025robust,
  title   = {Robust Time-Varying Control Barrier Functions with Sector-Bounded Nonlinearities},
  author  = {Biert{\"u}mpfel, Felix and Chun, Jungbae and Seiler, Peter},
  journal = {arXiv preprint arXiv:2511.09784},
  year    = {2025},
  url     = {https://arxiv.org/abs/2511.09784}
}

@book{khalil2002nonlinear,
  title     = {Nonlinear Systems},
  author    = {Khalil, Hassan K.},
  year      = {2002},
  edition   = {3rd},
  publisher = {Prentice Hall},
  address   = {Upper Saddle River, NJ}
}

@article{bhat2000finite,
  title     = {Finite-Time Stability of Continuous Autonomous Systems},
  author    = {Bhat, Sanjay P. and Bernstein, Dennis S.},
  journal   = {SIAM Journal on Control and Optimization},
  year      = {2000},
  publisher = {Society for Industrial and Applied Mathematics}
}

@inproceedings{raman2014model,
  title={Model predictive control with signal temporal logic specifications},
  author={Raman, Vasumathi and Donz{\'e}, Alexandre and Maasoumy, Mehdi and Murray, Richard M and Sangiovanni-Vincentelli, Alberto and Seshia, Sanjit A},
  booktitle={53rd IEEE Conference on Decision and Control},
  pages={81--87},
  year={2014},
  organization={IEEE}
}

@article{glotfelter2017nonsmooth,
  title={Nonsmooth barrier functions with applications to multi-robot systems},
  author={Glotfelter, Paul and Cort{\'e}s, Jorge and Egerstedt, Magnus},
  journal={IEEE Control Systems Letters},
  volume={1},
  number={2},
  pages={310--315},
  year={2017},
  publisher={IEEE}
}

@article{molnar2022safety,
  title={Safety-critical control with input delay in dynamic environment},
  author={Molnar, Tamas G and Kiss, Adam K and Ames, Aaron D and Orosz, G{\'a}bor},
  journal={IEEE transactions on control systems technology},
  volume={31},
  number={4},
  pages={1507--1520},
  year={2022},
  publisher={IEEE}
}

@article{lindemann2020barrier,
  title={Barrier function based collaborative control of multiple robots under signal temporal logic tasks},
  author={Lindemann, Lars and Dimarogonas, Dimos V},
  journal={IEEE Transactions on Control of Network Systems},
  volume={7},
  number={4},
  pages={1916--1928},
  year={2020},
  publisher={IEEE}
}

@article{garg2022fixed,
  title={Fixed-time control under spatiotemporal and input constraints: A {Quadratic Programming} based approach},
  author={Garg, Kunal and Arabi, Ehsan and Panagou, Dimitra},
  journal={Automatica},
  volume={141},
  pages={110314},
  year={2022},
  publisher={Elsevier}
}

@inproceedings{taylor2022safety,
  title={Safety of sampled-data systems with control barrier functions via approximate discrete time models},
  author={Taylor, Andrew J and Dorobantu, Victor D and Cosner, Ryan K and Yue, Yisong and Ames, Aaron D},
  booktitle={2022 IEEE 61st Conference on Decision and Control (CDC)},
  pages={7127--7134},
  year={2022},
  organization={IEEE}
}

@inproceedings{singletary2020control,
  title={Control barrier functions for sampled-data systems with input delays},
  author={Singletary, Andrew and Chen, Yuxiao and Ames, Aaron D},
  booktitle={2020 59th IEEE Conference on Decision and Control (CDC)},
  pages={804--809},
  year={2020},
  organization={IEEE}
}

\end{document}